\documentclass[english,12pt]{article}
\usepackage[T1]{fontenc}
\usepackage[latin9]{inputenc}
\usepackage{geometry}
\usepackage{color}
\usepackage{babel}
\usepackage{amsmath}
\usepackage{amsthm}
\usepackage{amssymb}
\usepackage{graphicx}
\usepackage{setspace}
\usepackage[authoryear]{natbib}
\usepackage[unicode=true,
 bookmarks=true,bookmarksnumbered=false,bookmarksopen=false,
 breaklinks=false,pdfborder={0 0 0},pdfborderstyle={},backref=false,colorlinks=true]
 {hyperref}

\makeatletter
\newcommand{\lyxaddress}[1]{
\par {\raggedright #1
\vspace{1.4em}
\noindent\par}
}
  \theoremstyle{plain}
  \newtheorem{lem}{\protect\lemmaname}
  \theoremstyle{definition}
  \newtheorem{defn}{\protect\definitionname}
  \theoremstyle{plain}
  \newtheorem{thm}{\protect\theoremname}
 \ifx\proof\undefined\
   \newenvironment{proof}[1][\proofname]{\par
     \normalfont\topsep6\p@\@plus6\p@\relax
     \trivlist
     \itemindent\parindent
     \item[\hskip\labelsep
           \scshape
       #1]\ignorespaces
   }{%
     \endtrivlist\@endpefalse
   }
   \providecommand{\proofname}{Proof}
 \fi

\setcitestyle{round}
\date{}

\makeatother

\providecommand{\definitionname}{Definition}
\providecommand{\lemmaname}{Lemma}
\providecommand{\theoremname}{Theorem}

\begin{document}

\title{\textbf{Finite-Sample Bounds for the Multivariate Behrens--Fisher
Distribution with Proportional Covariances}}

\author{Yixuan Qiu and Lingsong Zhang}
\maketitle

\lyxaddress{\begin{center}
Department of Statistics, Purdue University, West Lafayette, Indiana
47907, U.S.A.
\par\end{center}}
\begin{abstract}
The Behrens--Fisher problem is a well-known hypothesis testing problem
in statistics concerning two-sample mean comparison. In this article,
we confirm one conjecture in \citet{eaton1972random}, which provides
stochastic bounds for the multivariate Behrens--Fisher test statistic
under the null hypothesis. We also extend their results on the stochastic
ordering of random quotients to the arbitrary finite dimensional case.
This work can also be seen as a generalization of \citet{hsu1938contribution}
that provided the bounds for the univariate Behrens--Fisher problem.
The results obtained in this article can be used to derive a testing
procedure for the multivariate Behrens--Fisher problem that strongly
controls the Type I error.

\noindent \textbf{Keywords}: Behrens--Fisher problem; Hypothesis
testing; Mean comparison; Stochastic bound; Type I error.
\end{abstract}

\section{Introduction}

\label{sec:introduction}

The Behrens--Fisher problem is one of the most well-known hypothesis
testing problems that has been extensively studied by many statisticians,
partly due to its simple form and numerous real applications. The
univariate Behrens--Fisher problem can be phrased as follows: Let
$X=(X_{1},\ldots,X_{m})$ and $Y=(Y_{1},\ldots,Y_{n})$ be two independent
random samples with $X_{i}\overset{iid}{\sim}N(\mu_{1},\sigma_{1}^{2})$
and $Y_{i}\overset{iid}{\sim}N(\mu_{2},\sigma_{2}^{2})$, where all
the four parameters are unknown, and the target is to test $H_{0}:\mu_{1}=\mu_{2}$
versus $H_{a}:\mu_{1}\ne\mu_{2}$.

There was enormous research on designing a procedure to test this
hypothesis, for example Fisher's fiducial inference \citep{fisher1935fiducial},
Scheff{\'e}'s $t$-distribution method \citep{scheffe1943solutions},
the generalized $p$-value method \citep{tsui1989generalized}, the
marginal inferential models \citep{martin2014marginal}, and many
others that were summarized in review articles such as \citet{scheffe1970practical}
and \citet{kim1998behrens}. Among all these approaches, the most
broadly-adopted test statistic is the Behrens--Fisher statistic.
Using conventional notations, $\overline{X}$ and $\overline{Y}$
are the two sample means, $S_{1}^{2}$ and $S_{2}^{2}$ are the two
unbiased sample variances, and then the Behrens--Fisher statistic
is defined by $T=\left(S_{1}^{2}/m+S_{2}^{2}/n\right)^{-1/2}(\overline{X}-\overline{Y})$.
It is well known that the sampling distribution of $T$ under $H_{0}$
depends on the unknown variance ratio $\sigma_{1}^{2}/\sigma_{2}^{2}$,
and various methods were proposed to approximate this null distribution,
for example the most widely-used Welch-Satterthwaite approximate degrees
of freedom \citep{satterthwaite1946approximate,welch1947generalization}.

Despite their extreme popularity in applications, one critical issue
of the approximation methods is that they do not guarantee the control
of Type I error. Therefore, conservative test procedures that can
strongly control the Type I error are also of interest. The remarkable
works \citet{hsu1938contribution} and \citet{mickey1966bounds} showed
that the distribution function of $T$ is bounded below by $t_{\min\{m-1,n-1\}}$,
the $t$-distribution with $\min\{m-1,n-1\}$ degrees of freedom,
and bounded above by $t_{m+n-2}$. With this result, one can use critical
values or $p$-values based on $t_{\min\{m-1,n-1\}}$ to test the
hypothesis, which ensures the limit of Type I error. This approach
also motivated works such as \citet{hayter2013new} and \citet{martin2014marginal}.

The Behrens--Fisher problem was also generalized to the multivariate
case in various research articles. In this setting, each observation
follows a multivariate normal distribution, and the target is to test
the equality of the two mean vectors. In the multivariate case, most
of the approaches are based on the approximate degrees of freedom
framework, for example \citet{yao1965approximate}, \citet{johansen1980welch},
\citet{nel1986solution}, and \citet{krishnamoorthy2004modified}.
Also see \citet{christensen1997comparison} for a comparison of other
solutions.

Alternatively, along the direction of \citet{hsu1938contribution}
and \citet{mickey1966bounds}, \citet{eaton1972random} attempted
to develop stochastic bounds for the test statistic in the multivariate
case, and they provided the result for the two-dimensional case with
proportional covariances assumption. However, the theorem that they
developed to prove the result had the restriction that it only applied
to the two-dimensional case, so they left the general finite dimensional
case as a conjecture.

In this article, we study the same problem as in \citet{eaton1972random}
using a related but different approach, and we are able to confirm
this conjecture and generalize their result to the arbitrary finite
dimensional case. As a result, we provide sharp bounds for the multivariate
Behrens--Fisher distribution with proportional covariances, as a
direct generalization of Hsu's result in the univariate case.

The remaining part of this article is organized as follows. In Section
\ref{sec:mbf} we briefly introduce the multivariate Behrens--Fisher
problem and review some existing results on it. Section \ref{sec:main_results}
is the main part of this article, where two major theorems that describe
the stochastic bounds for the test statistic are provided. In Section
\ref{sec:simulation} we use numerical simulations to illustrate the
performance of the proposed test compared with other approximation
methods. And finally in Section \ref{sec:conclusion}, some discussions
and the conclusion of this article are provided. The proofs of two
important lemmas are in the appendix.

\section{Multivariate Behrens--Fisher Problem}

\label{sec:mbf}

In this section we briefly describe the multivariate Behrens--Fisher
problem and review some relevant results on it. Similar to the univariate
case, let $X=(X_{1},\ldots,X_{m})^{\mathrm{T}}$ and $Y=(Y_{1},\ldots,Y_{n})^{\mathrm{T}}$
be two independent random samples, with each observation following
a $p$-dimensional multivariate normal distribution: $X_{i}\overset{iid}{\sim}N(\mu_{1},\Sigma_{1})$,
and $Y_{i}\overset{iid}{\sim}N(\mu_{2},\Sigma_{2})$. The problem
of interest is to test $H_{0}:\mu_{1}=\mu_{2}$ versus $H_{a}:\mu_{1}\ne\mu_{2}$,
with all the distributional parameters unknown. Following the same
assumption in \citet{eaton1972random}, we assume that $X$ and $Y$
have proportional covariances, \emph{i.e.},
\begin{equation}
\Sigma_{1}=\Sigma,\quad\Sigma_{2}=k\Sigma\label{eq:proportional_covariance}
\end{equation}
for some unknown $p\times p$ positive definite matrix $\Sigma$ and
an unknown constant $k$ ($k>0$). In the remaining part of this article
we assume that $p<\min\{m,n\}$.

Let $\overline{X}=m^{-1}\sum_{i=1}^{m}X_{i}$ and $\overline{Y}=n^{-1}\sum_{i=1}^{n}Y_{i}$
be the sample means, and $S_{1}=(m-1)^{-1}\sum_{i=1}^{m}(X_{i}-\overline{X})(X_{i}-\overline{X})^{\mathrm{T}}$
and $S_{2}=(n-1)^{-1}\sum_{i=1}^{n}(Y_{i}-\overline{Y})(Y_{i}-\overline{Y})^{\mathrm{T}}$
be the sample covariance matrices. It is well known that $\overline{X}\sim N(\mu_{1},m^{-1}\Sigma_{1}),\overline{Y}\sim N(\mu_{2},n^{-1}\Sigma_{2}),(m-1)S_{1}\sim W(\Sigma_{1},m-1)$,
and $(n-1)S_{2}\sim W(\Sigma_{2},n-1)$, where $W(\Sigma,n)$ stands
for a Wishart distribution with parameter $\Sigma$ and $n$ degrees
of freedom. All these four random vectors and matrices are independent
of each other. Furthermore, the multivariate Behrens--Fisher test
statistic is defined as
\begin{equation}
T^{2}=(\overline{X}-\overline{Y})^{\mathrm{T}}\left(m^{-1}S_{1}+n^{-1}S_{2}\right)^{-1}(\overline{X}-\overline{Y}),\label{eq:T2_statistic}
\end{equation}
and the sampling distribution of $T^{2}$ under $H_{0}$ is typically
called the multivariate Behrens--Fisher distribution. In this article,
our primary goal is to derive stochastic bounds for $T^{2}$ that
are free of the unknown parameters.

A major progress on this direction was made by \citet{eaton1972random}.
They first showed that under $H_{0}$,
\begin{equation}
T^{2}\overset{d}{=}Z^{\mathrm{T}}\left\{ \lambda(m-1)^{-1}W_{1}+(1-\lambda)(n-1)^{-1}W_{2}\right\} ^{-1}Z,\label{eq:T2_distribution}
\end{equation}
where $X\overset{d}{=}Y$ means $X$ and $Y$ have the same distribution,
$Z\sim N(0,I_{p}),W_{1}\sim W(I_{p},m-1),W_{2}\sim W(I_{p},n-1),\lambda=m^{-1}(m^{-1}+kn^{-1})^{-1}$,
$I_{p}$ is the $p\times p$ identity matrix, and $Z,W_{1}$, and
$W_{2}$ are independent. Then they proved that for $p=2$,
\begin{equation}
Z^{\mathrm{T}}\left\{ (m+n-2)^{-1}W_{(m+n-2)}\right\} ^{-1}Z\preceq_{st}T^{2}\preceq_{st}Z^{\mathrm{T}}\left\{ \nu^{-1}W_{(\nu)}\right\} ^{-1}Z,\label{eq:T2_stochastic_order}
\end{equation}
where $\preceq_{st}$ stands for the stochastic ordering, $\nu$ is
any integer satisfying $p\le\nu\le\min\{m-1,n-1\}$, and $W_{(n)}$
stands for a $W(I_{p},n)$ random matrix that is independent of $Z$.

However, in \citet{eaton1972random}, (\ref{eq:T2_stochastic_order})
was only proved for the case of $p=2$, since the underlying theory
did not generalize to higher dimensions. To overcome this difficulty,
in this article we use a different set of techniques to prove that
(\ref{eq:T2_stochastic_order}) also holds for $p>2$. The main results
are presented in Section \ref{sec:main_results}.

\section{Main Results}

\label{sec:main_results}

We first present two lemmas that are the keys to our main theorems,
whose proofs are given in the appendix. Lemma \ref{lem:linear_combination_z2}
studies the property of a linear combination of $Z_{i}^{2}$ random
variables, where $Z_{i}\overset{iid}{\sim}N(0,1)$.

\renewcommand{\labelenumi}{\alph{enumi})}
\begin{lem}
\label{lem:linear_combination_z2}Assume that $Z_{1},\ldots,Z_{p}$
are $p$ independent $N(0,1)$ random variables. Let $F(t;\theta)=F(t;\theta_{1},\ldots,\theta_{p})$
denote the distribution function of the random variable $T_{\theta}=\sum_{i=1}^{p}Z_{i}^{2}/\theta_{i},\theta_{i}>0$,
and define its partial derivatives as $f_{i}(t;\theta)=\partial F(t;\theta)/\partial\theta_{i}$
and $g_{i}(t;\theta)=\partial^{2}F(t;\theta)/\partial\theta_{i}^{2}$.
Then for $i,j=1,\ldots,p$, we have
\begin{enumerate}
\item $f_{i}(t;\theta)>0$,
\item $g_{i}(t;\theta)<0$, and
\item If $\theta_{i}<\theta_{j}$ then $f_{i}(t;\theta)>f_{j}(t;\theta)$.
\end{enumerate}
\end{lem}
Lemma \ref{lem:linear_combination_z2} itself gives some general properties
of the distribution family represented by $T_{\theta}$, and in this
article the lemma is mainly used to show the conclusion below, which
is the central technical tool to prove our main theorems.
\begin{lem}
\label{lem:h_concavity}Assume that $Z$ is a $N(0,I_{p})$ random
vector. Fix $t>0$ and let $M_{1}$ and $M_{2}$ be two $p\times p$
positive definite matrices. Define $h(\lambda;t,M_{1},M_{2})=\mathrm{pr}\left[Z^{\mathrm{T}}\left\{ \lambda M_{1}+(1-\lambda)M_{2}\right\} ^{-1}Z\le t\right]$,
and then $\partial^{2}h(\lambda;t,M_{1},M_{2})/\partial\lambda^{2}<0$.
\end{lem}
To present the main theorems of this article, we first introduce two
useful concepts: the majorization of vectors \citep{olkin2016inequalities},
and the exchangeability of a sequence of random vectors.
\begin{defn}
Let $x=(x_{1},\ldots,x_{n})$ and $y=(y_{1},\ldots,y_{n})$ be two
vectors in $\mathbb{R}^{n}$, and let $x_{[1]}\ge\cdots\ge x_{[n]}$
and $y_{[1]}\ge\cdots\ge y_{[n]}$ be the decreasing rearrangement
of $x$ and $y$ respectively. $x$ is said to be majorized by $y$,
denoted by $x\prec_{m}y$, if
\[
\sum_{i=1}^{n}x_{[i]}=\sum_{i=1}^{n}y_{[i]}\text{ and }\sum_{i=1}^{k}x_{[i]}\le\sum_{i=1}^{k}y_{[i]},k=1,2,\ldots,n-1.
\]
\end{defn}
Intuitively, $x\prec_{m}y$ indicates that $x$ and $y$ have the
same total quantity, but $y$ is more ``spread out'', or less ``equally
allocated'' than $x$.
\begin{defn}
A sequence of random vectors $W=(W_{1},\ldots,W_{r})$ is said to
be exchangeable, if for any permutation $\pi$ of $(1,2,\ldots,r)$,
$(W_{\pi(1)},\ldots,W_{\pi(r)})\overset{d}{=}(W_{1},\ldots,W_{r})$.
\end{defn}
With these notations, the first main result of this article is summarized
in Theorem \ref{thm:main_theorem}.
\begin{thm}
\label{thm:main_theorem}Let $W=(W_{1},\ldots,W_{r})$ be an exchangeable
sequence of positive definite random matrices of size $p\times p$,
and let $Z$ be a $N(0,I_{p})$ random vector that is independent
of $W$. If $\psi=(\psi_{1},\ldots,\psi_{r})$ and $\eta=(\eta_{1},\ldots,\eta_{r})$
are two sequences of nonnegative constants such that $\psi\prec_{m}\eta$,
then
\[
Z^{\mathrm{T}}\left(\sum_{i=1}^{r}\psi_{i}W_{i}\right)^{-1}Z\preceq_{st}Z^{\mathrm{T}}\left(\sum_{i=1}^{r}\eta_{i}W_{i}\right)^{-1}Z.
\]
\end{thm}
\begin{proof}
Let $\mathbb{P}_{p}$ denote the space of all $p\times p$ positive
definite matrices. Fix $t>0$, and define the function $\phi:(\mathbb{P}_{p})^{r}\mapsto\mathbb{R}$,
$\phi(w_{1},\ldots,w_{r})=1-\mathrm{pr}\left\{ Z^{\mathrm{T}}\left(\sum_{i=1}^{r}w_{i}\right)^{-1}Z\le t\right\} $
with each $w_{i}\in\mathbb{P}_{p}$. We are going to show that $\phi$
is convex, \emph{i.e.}, given $u_{1},\ldots,u_{r}\in\mathbb{P}_{p}$,
$v_{1},\ldots,v_{r}\in\mathbb{P}_{p}$, and any constant $0\le\lambda\le1$,
$\phi$ satisfies 
\begin{equation}
\phi(\lambda u_{1}+(1-\lambda)v_{1},\ldots,\lambda u_{r}+(1-\lambda)v_{r})\le\lambda\phi(u_{1},\ldots,u_{r})+(1-\lambda)\phi(v_{1},\ldots,v_{r}).\label{eq:phi_convexity}
\end{equation}

To verify this, let $\tilde{u}=\sum_{i=1}^{r}u_{i}\in\mathbb{P}_{p}$
and $\tilde{v}=\sum_{i=1}^{r}v_{i}\in\mathbb{P}_{p}$, so $\phi(\lambda u_{1}+(1-\lambda)v_{1},\ldots,\lambda u_{r}+(1-\lambda)v_{r})=1-\mathrm{pr}\left[Z^{\mathrm{T}}\left\{ \lambda\tilde{u}+(1-\lambda)\tilde{v}\right\} ^{-1}Z\le t\right]=1-h(\lambda;t,\tilde{u},\tilde{v})$,
where $h(\cdot)$ is defined in Lemma \ref{lem:h_concavity}. It follows
from Lemma \ref{lem:h_concavity} that $h(\cdot)$ is concave, implying
$h(\lambda)\ge\lambda h(1)+(1-\lambda)h(0)$. Since $h(1)=1-\phi(u_{1},\ldots,u_{r})$
and $h(0)=1-\phi(v_{1},\ldots,v_{r})$, (\ref{eq:phi_convexity})
holds immediately.

Moreover, $\phi$ is continuous and exchangeable on its arguments,
so it satisfies the condition of Theorem 2.4 of \citet{eaton1972random}.
As a consequence of the theorem, it follows that
\[
E_{W}\left\{ \phi(\psi_{1}W_{1},\ldots,\psi_{r}W_{r})\right\} \le E_{W}\left\{ \phi(\eta_{1}W_{1},\ldots,\eta_{r}W_{r})\right\} ,
\]
where the expectation is taken on the joint distribution of $W=(W_{1},\ldots,W_{r})$.
It is easy to see that $E_{W}\left\{ \phi(\psi_{1}W_{1},\ldots,\psi_{r}W_{r})\right\} =1-\mathrm{pr}\left\{ Z^{\mathrm{T}}\left(\sum_{i=1}^{r}\psi_{i}W_{i}\right)^{-1}Z\le t\right\} $,
so we obtain $\mathrm{pr}\left\{ Z^{\mathrm{T}}\left(\sum_{i=1}^{r}\psi_{i}W_{i}\right)^{-1}Z\le t\right\} \ge\mathrm{pr}\left\{ Z^{\mathrm{T}}\left(\sum_{i=1}^{r}\eta_{i}W_{i}\right)^{-1}Z\le t\right\} $
for any $t>0$, which concludes the proof.
\end{proof}
Theorem \ref{thm:main_theorem} does not put any specific distributional
assumptions on $W$, so it is more general than the Behrens--Fisher
problem setting where $W_{i}$'s follow Wishart distributions. Applying
Theorem \ref{thm:main_theorem} to the multivariate Behrens--Fisher
test statistic in (\ref{eq:T2_statistic}), we obtain the following
result:
\begin{thm}
\label{thm:mbf}For the multivariate Behrens--Fisher test statistic
(\ref{eq:T2_statistic}), under $H_{0}$ and the proportional covariances
assumption (\ref{eq:proportional_covariance}),
\begin{equation}
F_{p,\min\{m,n\}-p}\left(\frac{\min\{m,n\}-p}{p(\min\{m,n\}-1)}t\right)\le P(T^{2}\le t)\le F_{p,m+n-p-1}\left(\frac{m+n-p-1}{p(m+n-2)}t\right),\label{eq:T2_stochastic_order_2}
\end{equation}
where $F_{a,b}(\cdot)$ is the distribution function of an $F$-distribution
with $a$ and $b$ degrees of freedom.
\end{thm}
\begin{proof}
We first show that equation (\ref{eq:T2_distribution}) holds. Under
$H_{0}$, $\overline{X}-\overline{Y}\sim N(0,(m^{-1}+kn^{-1})\Sigma)$.
Let $c=m^{-1}+kn^{-1},\Delta=c\Sigma$, and $\Delta^{1/2}$ is the
symmetric square root of $\Delta$, then $T^{2}=(\overline{X}-\overline{Y})^{\mathrm{T}}\Delta^{-1/2}(m^{-1}\Delta^{-1/2}S_{1}\Delta^{-1/2}+n^{-1}\Delta^{-1/2}S_{2}\Delta^{-1/2})^{-1}\Delta^{-1/2}(\overline{X}-\overline{Y})$.
It follows that $\Delta^{-1/2}(\overline{X}-\overline{Y})\sim N(0,I_{p})$,
$\Delta^{-1/2}S_{1}\Delta^{-1/2}\sim W((m-1)^{-1}c^{-1}I_{p},m-1)$,
and $\Delta^{-1/2}S_{2}\Delta^{-1/2}\sim W(k(n-1)^{-1}c^{-1}I_{p},n-1)$.
If we let $\lambda=m^{-1}(m^{-1}+kn^{-1})^{-1}$, then (\ref{eq:T2_distribution})
can be obtained immediately.

Now let $\nu=m-1,\theta=n-1$, $\psi=\left((\nu+\theta)^{-1},\ldots,(\nu+\theta)^{-1}\right)\in\mathbb{R}^{\nu+\theta}$,
$\eta$ be a vector that contains $\nu$ elements of $\lambda/\nu$
and $\theta$ elements of $(1-\lambda)/\theta$, and $\xi\in\mathbb{R}^{\nu+\theta}$
be a vector that contains $\min\{\nu,\theta\}$ elements of $\min\{\nu,\theta\}^{-1}$
and other elements equal to zero, then it is easy to verify that $\psi\prec_{m}\eta\prec_{m}\xi$.
According to Theorem \ref{thm:main_theorem}, we have
\[
Z^{\mathrm{T}}\left(\sum_{i=1}^{r}\psi_{i}W_{i}\right)^{-1}Z\preceq_{st}Z^{\mathrm{T}}\left(\sum_{i=1}^{r}\eta_{i}W_{i}\right)^{-1}Z\preceq_{st}Z^{\mathrm{T}}\left(\sum_{i=1}^{r}\xi_{i}W_{i}\right)^{-1}Z.
\]
It is obvious that $\sum_{i=1}^{r}\psi_{i}W_{i}\sim W((\nu+\theta)^{-1}I_{p},\nu+\theta)$,
$T^{2}\overset{d}{=}Z^{\mathrm{T}}\left(\sum_{i=1}^{r}\eta_{i}W_{i}\right)^{-1}Z$,
and $\sum_{i=1}^{r}\xi_{i}W_{i}\sim W(\min\{\nu,\theta\}^{-1}I_{p},\min\{\nu,\theta\})$.
Combining with the fact that $nZ^{\mathrm{T}}W^{-1}Z\overset{d}{=}np(n-p+1)^{-1}F$
\citep{rao1973linear}, where $W\sim W(I_{p},n)$ and $F\sim F_{p,n-p+1}$,
(\ref{eq:T2_stochastic_order_2}) is confirmed.
\end{proof}
Using the inequality in Theorem \ref{thm:mbf}, a $p$-value of the
test can be computed as
\begin{equation}
\textsc{pval}(T^{2})=1-F_{p,\min\{m,n\}-p}\left(\frac{\min\{m,n\}-p}{p(\min\{m,n\}-1)}T^{2}\right),\label{eq:pvalue}
\end{equation}
and it is guaranteed that under $H_{0}$, $\mathrm{pr}\{\textsc{pval}(T^{2})\le\alpha\}\le\alpha$.

\section{Simulation Study}

\label{sec:simulation}

In this section we conduct simulation experiments to compare the testing
procedure using (\ref{eq:pvalue}) with other existing methods, including
\citet{yao1965approximate}, \citet{johansen1980welch}, \citet{nel1986solution},
and \citet{krishnamoorthy2004modified}, in terms of their Type I
errors. The experiment setting is as follows. We fix the number of
variables $p=5$, and assume that $m\le n$ without loss of generality.
Two groups of sample sizes are considered: the ``small sample''
group, with $m=10$ and $n=10,20,50$; and the ``large sample''
group, with $m=100$ and $n=100,200,500$. The true $\Sigma$ is a
realization of the $W(I_{5},10)$ distribution, and its value is fixed
during the experiment. Five different values of $k$ are considered,
$k=0.01,0.1,1,10,100$, for each combination of $m$ and $n$. Then
for each parameter setting of $(m,n,k)$, the data $(X,Y)$ are randomly
sampled 100,000 times to compute the empirical Type I error for each
method.

Figure \ref{fig:error_alpha_0.05} illustrates the results for significance
level $\alpha=0.05$. The first four methods correspond to the the
existing solutions, and the ``F-Bound'' method is the one based
on (\ref{eq:pvalue}). As can be seen from the last three columns
of the plot matrix, which correspond to the ``large sample'' case,
all five solutions perform reasonably well. However, when the sample
sizes are small, as in the first three columns of the plot matrix,
the existing methods tend to exaggerate the Type I error a lot, and
even double the pre-specified significance level in some situations.
On the contrary, even if the F-Bound method is conservative in worst
cases, it always guarantees the control of Type I error.

\begin{figure}
\begin{centering}
\includegraphics[width=0.85\textwidth]{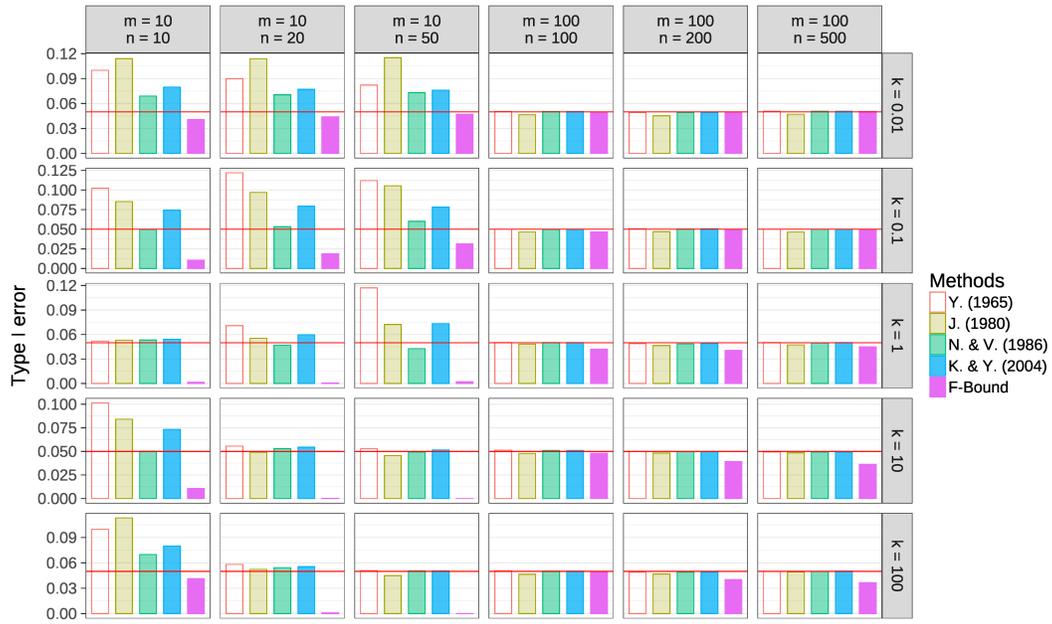}
\par\end{centering}
\caption{\label{fig:error_alpha_0.05}Type I errors of five testing methods
for the multivariate Behrens--Fisher problem under different parameter
settings, with each setting displayed in one sub-plot. The significance
level is set to $\alpha=0.05$, indicated by the horizontal lines
in each sub-plot, and the height of the bars stands for the Type I
error. The first four methods are existing solutions to the problem,
and the one labeled with ``F-Bound'' is the approach based on (\ref{eq:pvalue}).}

\end{figure}

This phenomenon is even more clear under the $\alpha=0.01$ situation,
as is shown in Figure \ref{fig:error_alpha_0.01}. Under some circumstances
the existing methods inflate the Type I error more than four times,
which may cause unreliable conclusions in real applications. Same
as the previous case, the F-Bound method is always valid despite its
conservativeness.

\begin{figure}
\begin{centering}
\includegraphics[width=0.85\textwidth]{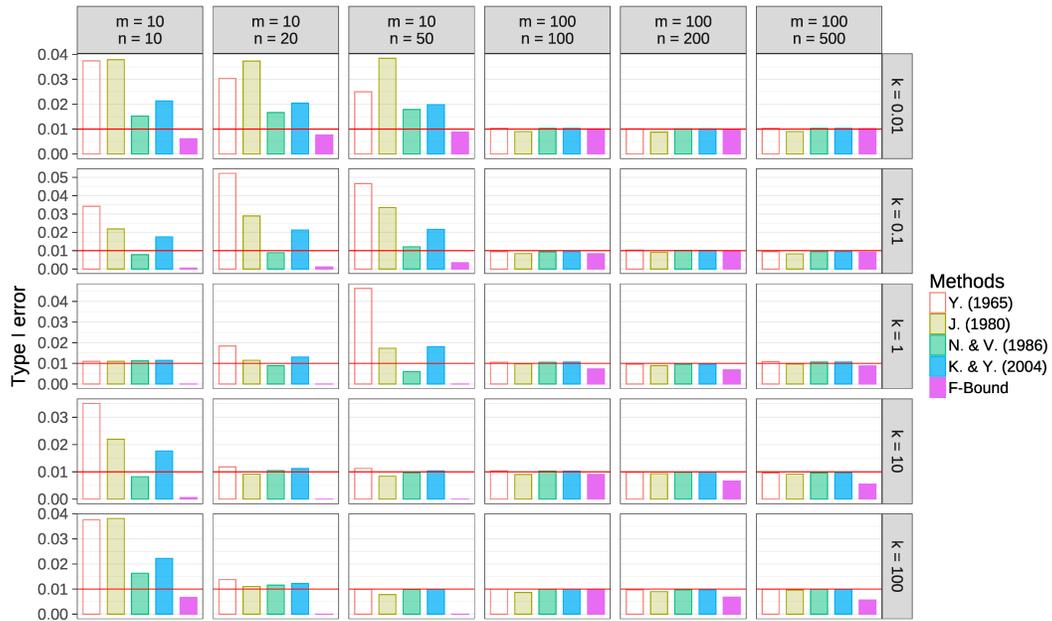}
\par\end{centering}
\caption{\label{fig:error_alpha_0.01}The same plot as Figure \ref{fig:error_alpha_0.05}
but with significance level $\alpha=0.01$.}
\end{figure}

To summarize, the simulation study indicates that the theoretical
result obtained in this article is useful to derive a testing procedure
for the multivariate Behrens--Fisher problem that guarantees the
Type I error control, which is crucial for many scientific studies.

\section{Discussion and Conclusion}

\label{sec:conclusion}

In this article we have revisited the multivariate Behrens--Fisher
problem with the proportional covariances assumption, and have derived
finite-sample lower and upper bounds for the null distribution of
the test statistic. This result extends the previous work by \citet{hsu1938contribution}
for the univariate case and \citet{eaton1972random} for the two-dimensional
case, and can be used to create a testing procedure that strongly
controls the Type I error for the multivariate Behrens--Fisher problem.

It is true that the proportional covariances assumption (\ref{eq:proportional_covariance})
is a moderately strong restriction, and one may hope to verify the
result for the most general forms of $\Sigma_{1}$ and $\Sigma_{2}$.
In this article, this assumption is made based on the following two
considerations. First, the original motivation of this article was
to generalize Theorem 3.1 of \citet{eaton1972random}, about the stochastic
ordering of a series of random quotients, from two-dimension to any
finite dimension. However, the test statistic for the most general
Behrens--Fisher problem does not belong to this type of random quotient.
Second, the technical difficulty of the general case is expected to
be formidable. As can be seen from Lemma \ref{lem:h_concavity}, there
exists some concavity property for the proportional covariances case,
which greatly helps proving the bounds. However, many examples can
be given to show that such properties are totally destroyed in the
general case, so some more advanced techniques need to be developed
in order to fully solve the general situation. We leave this possibility
for future research.

\appendix

\section{Appendix}

\subsection{Proof of Lemma \ref{lem:linear_combination_z2}}
\begin{proof}
Since $\theta_{1},\ldots,\theta_{p}$ are exchangeable in $F(t;\theta)$,
we will prove the case for $i=1,j=2$ without loss of generality.
Define the random variable $T_{12}=Z_{1}^{2}/\theta_{1}+Z_{2}^{2}/\theta_{2}$
with the distribution function $F_{12}(t;\theta_{1},\theta_{2})$,
and let $\phi(x)$ and $\Phi(x)$ denote the density function and
distribution function of $N(0,1)$, respectively, then
\begin{align}
F_{12}(t;\theta_{1},\theta_{2}) & =\int_{-(\theta_{2}t)^{1/2}}^{(\theta_{2}t)^{1/2}}\mathrm{pr}(Z_{1}^{2}/\theta_{1}+s^{2}/\theta_{2}\le t)\phi(s)\mathrm{d}s\nonumber \\
 & =2\int_{0}^{(\theta_{2}t)^{1/2}}\left\{ 2\Phi\left(\left\{ \theta_{1}(t-s^{2}/\theta_{2})\right\} ^{1/2}\right)-1\right\} \phi(s)\mathrm{d}s\nonumber \\
 & =4\int_{0}^{(\theta_{2}t)^{1/2}}\Phi\left(\left\{ \theta_{1}(t-s^{2}/\theta_{2})\right\} ^{1/2}\right)\phi(s)\mathrm{d}s-2\Phi((\theta_{2}t)^{1/2})+1,\nonumber \\
\frac{\partial F_{12}(t;\theta_{1},\theta_{2})}{\partial\theta_{1}} & =2\int_{0}^{(\theta_{2}t)^{1/2}}(t-s^{2}/\theta_{2})^{1/2}\theta_{1}^{-1/2}\phi\left(\left\{ \theta_{1}(t-s^{2}/\theta_{2})\right\} ^{1/2}\right)\phi(s)\mathrm{d}s>0.\label{eq:dF12_d1}
\end{align}
Moreover, using the fact that $\phi'(x)=-x\phi(x)$, we have
\begin{align*}
 & \frac{\partial^{2}F_{12}(t;\theta_{1},\theta_{2})}{\partial\theta_{1}^{2}}\\
= & -\int_{0}^{(\theta_{2}t)^{1/2}}\left\{ \left(t-s^{2}\theta_{2}^{-1}\right)^{3/2}\theta_{1}^{-1/2}+\left(t-s^{2}\theta_{2}^{-1}\right)^{\frac{1}{2}}\theta_{1}^{-3/2}\right\} \phi\left(\left\{ \theta_{1}(t-s^{2}\theta_{2}^{-1})\right\} ^{1/2}\right)\phi(s)\mathrm{d}s<0.
\end{align*}

Let $\tilde{f}(t;\theta_{3},\ldots,\theta_{p})$ be the density function
of $\tilde{T}=\sum_{i=3}^{p}Z_{i}^{2}/\theta_{i}$, and then $T_{\theta}=T_{12}+\tilde{T}$
has the distribution function $F(t;\theta)=\int_{0}^{t}F_{12}(s;\theta_{1},\theta_{2})\tilde{f}(t-s;\theta_{3},\ldots,\theta_{p})\mathrm{d}s$.
Taking the partial derivatives with respect to $\theta_{1}$ on both
sides, we have $f_{1}(t;\theta)=\int_{0}^{t}\left(\partial F_{12}(s;\theta_{1},\theta_{2})/\partial\theta_{1}\right)\tilde{f}(t-s;\theta_{3},\ldots,\theta_{p})\mathrm{d}s>0$
and $g_{1}(t;\theta)=\int_{0}^{t}\left(\partial^{2}F_{12}(s;\theta_{1},\theta_{2})/\partial\theta_{1}^{2}\right)\tilde{f}(t-s;\theta_{3},\ldots,\theta_{p})\mathrm{d}s<0$,
which prove the statements \emph{a)} and \emph{b)}.

Now let $h(\theta_{1},\theta_{2})=\partial F_{12}(t;\theta_{1},\theta_{2})/\partial\theta_{1}$
as in (\ref{eq:dF12_d1}), and fix $0<a<b$ with $r=b/a>1$. With
change of variables $u=s(bt)^{-1/2}$ followed by $\rho=\arcsin(u)$,
we obtain
\begin{align*}
h(a,b) & =2r^{1/2}t\int_{0}^{1}(1-u^{2})^{1/2}\phi((at)^{1/2}(1-u^{2})^{1/2})\phi((bt)^{1/2}u)\mathrm{d}u\\
 & =r^{1/2}t\pi^{-1}\int_{0}^{1}(1-u^{2})^{1/2}\exp\left\{ -at(1-u^{2})/2-btu^{2}/2\right\} \mathrm{d}u\\
 & =r^{1/2}t\pi^{-1}\int_{0}^{\pi/2}\cos^{2}\rho\exp(-at\cos^{2}\rho/2-bt\sin^{2}\rho/2)\mathrm{d}\rho\equiv r^{1/2}t\pi^{-1}I_{ab}.
\end{align*}
Similarly, by switching the order of $a$ and $b$ and with another
change of variable $\eta=\pi/2-\rho$, it follows that
\[
h(b,a)=r^{-1/2}t\pi^{-1}\int_{0}^{\pi/2}\sin^{2}\rho\exp(-at\cos^{2}\eta/2-bt\sin^{2}\eta/2)\mathrm{d}\eta\equiv r^{-1/2}t\pi^{-1}I_{ba}.
\]
Therefore,
\begin{align}
 & h(a,b)-h(b,a)\nonumber \\
 & \ge t\pi^{-1}(I_{ab}-I_{ba})=t\pi^{-1}\int_{0}^{\pi/2}\cos(2\rho)\exp(-at\cos^{2}\rho/2-bt\sin^{2}\rho/2)\mathrm{d}\rho\nonumber \\
 & =t\pi^{-1}\exp\{-(a+b)t/4\}\int_{0}^{\pi/2}\cos(2\rho)\exp\{(b-a)t\cos(2\rho)/4\}\mathrm{d}\rho\nonumber \\
 & =(b-a)t^{2}(8\pi)^{-1}\exp\{-(a+b)t/4\}\int_{0}^{\pi/2}\sin^{2}(2\rho)\exp\{(b-a)t\cos(2\rho)/4\}\mathrm{d}\rho>0.\label{eq:hab_hba}
\end{align}

Now for $f_{2}(t;\theta)$, let $\theta_{ab}=(a,b,\theta_{3},\ldots,\theta_{p})$
and $\theta_{ba}=(b,a,\theta_{3},\ldots,\theta_{p})$, and then by
symmetry we have $f_{2}(t;\theta_{ab})=f_{1}(t;\theta_{ba})$. Hence
as a consequence of (\ref{eq:hab_hba}), we finally get
\begin{align*}
f_{1}(t;\theta_{ab})-f_{2}(t;\theta_{ab}) & =f_{1}(t;\theta_{ab})-f_{1}(t;\theta_{ba})\\
 & =\int_{0}^{t}\{h(a,b)-h(b,a)\}\tilde{f}(t-s;\theta_{3},\ldots,\theta_{p})\mathrm{d}s>0,
\end{align*}
whenever $0<a<b$, which concludes the proof of \emph{c)}.
\end{proof}

\subsection{Proof of Lemma \ref{lem:h_concavity}}
\begin{proof}
For simplicity we omit the parameters $t,M_{1},$ and $M_{2}$ in
$h(\cdot)$ when no confusion is caused. Let $M(\lambda)=\lambda M_{1}+(1-\lambda)M_{2}$
be a matrix-valued function dependent on $\lambda$, and assume its
eigen decomposition is $M(\lambda)=\Gamma(\lambda)D(\lambda)\Gamma(\lambda)^{\mathrm{T}}$,
where $D(\lambda)=\mathrm{diag}\left\{ d_{1}(\lambda),\ldots,d_{p}(\lambda)\right\} $
contains the sorted eigenvalues $d_{1}(\lambda)\ge\cdots\ge d_{p}(\lambda)>0$,
and $\Gamma(\lambda)=(\gamma_{1}(\lambda),\ldots,\gamma_{p}(\lambda))$
are the associated eigenvectors. Again we will omit the $\lambda$
arguments in the relevant quantities above whenever appropriate.

Since $M^{-1}=\Gamma D^{-1}\Gamma^{\mathrm{T}}$, we have $h(\lambda)=\mathrm{pr}(Z^{\mathrm{T}}\Gamma D^{-1}\Gamma^{\mathrm{T}}Z\le t)=\mathrm{pr}(Z^{\mathrm{T}}D^{-1}Z\le t)$.
The second identity holds since $\Gamma^{\mathrm{T}}Z\sim N(0,\Gamma^{\mathrm{T}}\Gamma),\Gamma^{\mathrm{T}}\Gamma=I_{p}$
and thus $Z\overset{d}{=}\Gamma^{\mathrm{T}}Z$. Therefore, using
the notations in Lemma \ref{lem:linear_combination_z2}, we have $h(\lambda)=\mathrm{pr}(\sum_{i=1}^{p}Z_{i}^{2}/d_{i}\le t)=F(t;\delta)$
where $\delta=(d_{1},\ldots,d_{p})$. As a result,
\begin{equation}
h''(\lambda)=\sum_{i=1}^{p}\left[g_{i}(t;\delta)\left(\frac{\partial d_{i}}{\partial\lambda}\right)^{2}+f_{i}(t;\delta)\left(\frac{\partial^{2}d_{i}}{\partial\lambda^{2}}\right)\right],\label{eq:h_second_deriv}
\end{equation}
where $f_{i}(t;\delta)$ and $g_{i}(t;\delta)$ are also defined in
Lemma \ref{lem:linear_combination_z2}.

Theorem 9 and Theorem 10 of \citet{lancaster1964eigenvalues} provide
explicit expressions for $\partial^{2}d_{i}/\partial\lambda^{2}$,
where the former assumes $d_{i}$'s are distinct while the latter
considers multiplicity of eigenvalues. For now we shall assume that
$d_{i}$'s are all distinct for brevity of the proof. The same technique
applies to the more general case.

Let $M^{(k)}$ be the $k$th derivative of $M$ with respect to $\lambda$,
then clearly $M^{(1)}=M_{1}-M_{2}$ and $M^{(2)}=O$ where $O$ is
the zero matrix. Also define $p_{ij}=\gamma_{i}^{\mathrm{T}}M^{(1)}\gamma_{j}=p_{ji}$,
then according to Theorem 9 of \citet{lancaster1964eigenvalues},
\[
\frac{\partial^{2}d_{i}}{\partial\lambda^{2}}=2\sum_{\substack{k=1\\
k\neq i
}
}^{p}\frac{p_{ik}p_{ki}}{d_{i}-d_{k}}=2\sum_{\substack{k=1\\
k\neq i
}
}^{p}\frac{p_{ik}^{2}}{d_{i}-d_{k}}.
\]
Now consider the cumulative sum of eigenvalues from the bottom, defined
as $c_{i}=\sum_{j=i}^{p}d_{j}$, whose second derivative is given
by
\begin{align}
\frac{\partial^{2}c_{i}}{\partial\lambda^{2}} & =\sum_{j=i}^{p}\frac{\partial^{2}d_{j}}{\partial\lambda^{2}}=2\sum_{j=i}^{p}\sum_{\substack{k=1\\
k\neq j
}
}^{p}\frac{p_{jk}^{2}}{d_{j}-d_{k}}=2\sum_{j=i}^{p}\left(\sum_{k=1}^{i-1}\frac{p_{jk}^{2}}{d_{j}-d_{k}}+\sum_{\substack{k=i\\
k\neq j
}
}^{p}\frac{p_{jk}^{2}}{d_{j}-d_{k}}\right)\nonumber \\
 & =2\sum_{j=i}^{p}\sum_{k=1}^{i-1}\frac{p_{jk}^{2}}{d_{j}-d_{k}}+\sum_{\substack{j,k\ge i\\
j\neq k
}
}\frac{p_{jk}^{2}}{d_{j}-d_{k}}.\label{eq:summations}
\end{align}
For $j\neq k$, $p_{jk}^{2}(d_{j}-d_{k})^{-1}+p_{kj}^{2}(d_{k}-d_{j})^{-1}=0$,
so the second term in (\ref{eq:summations}) is zero. For the first
term, since $k<j$ and hence $d_{j}<d_{k}$, we conclude that $\partial^{2}c_{i}/\partial\lambda^{2}<0$.

With this result, $h''(\lambda)$ in (\ref{eq:h_second_deriv}) can
be written as
\begin{align}
h''(\lambda) & =\sum_{i=1}^{p}g_{i}(t;\delta)\left(\frac{\partial d_{i}}{\partial\lambda}\right)^{2}+\sum_{i=1}^{p}f_{i}(t;\delta)\left(\frac{\partial^{2}d_{i}}{\partial\lambda^{2}}\right)\nonumber \\
 & <\sum_{i=1}^{p}f_{i}(t;\delta)\left(\frac{\partial^{2}d_{i}}{\partial\lambda^{2}}\right)=\sum_{i=1}^{p}\tilde{f}_{i}(t;\delta)\left(\frac{\partial^{2}c_{i}}{\partial\lambda^{2}}\right),\label{eq:h_second_deriv_ineq}
\end{align}
where $\tilde{f}_{1}=f_{1}$ and $\tilde{f}_{i}=f_{i}-f_{i-1}$ for
$i\ge2$. The inequality in (\ref{eq:h_second_deriv_ineq}) holds
since $g_{i}(t;\delta)<0$ by part \emph{b)} of Lemma \ref{lem:linear_combination_z2}.
Moreover, since $d_{1}\ge\cdots\ge d_{p}$, we have $f_{1}\le\cdots\le f_{p}$
and thus $\tilde{f}_{i}\ge0$ by part \emph{a)} and \emph{c)} of Lemma
\ref{lem:linear_combination_z2}. This implies that $h''(\lambda)<0$
and hence concludes the proof. 
\end{proof}
\bibliographystyle{biometrika}
\nocite{*}
\bibliography{ref}

\begin{thebibliography}{24}
\expandafter\ifx\csname natexlab\endcsname\relax\def\natexlab#1{#1}\fi

\bibitem[{Christensen \& Rencher(1997)}]{christensen1997comparison}
\textsc{Christensen, W.~F.} \& \textsc{Rencher, A.~C.} (1997).
\newblock A comparison of type i error rates and power levels for seven
  solutions to the multivariate behrens-fisher problem.
\newblock \textit{Communications in Statistics-Simulation and Computation}
  \textbf{26}, 1251--1273.

\bibitem[{Eaton \& Olshen(1972)}]{eaton1972random}
\textsc{Eaton, M.~L.} \& \textsc{Olshen, R.~A.} (1972).
\newblock Random quotients and the behrens-fisher problem.
\newblock \textit{The Annals of Mathematical Statistics} \textbf{43},
  1852--1860.

\bibitem[{Fisher(1935)}]{fisher1935fiducial}
\textsc{Fisher, R.~A.} (1935).
\newblock The fiducial argument in statistical inference.
\newblock \textit{Annals of Human Genetics} \textbf{6}, 391--398.

\bibitem[{Hayter(2013)}]{hayter2013new}
\textsc{Hayter, A.} (2013).
\newblock A new procedure for the behrens--fisher problem that guarantees
  confidence levels.
\newblock \textit{Journal of Statistical Theory and Practice} \textbf{7},
  515--536.

\bibitem[{Hsu(1938)}]{hsu1938contribution}
\textsc{Hsu, P.} (1938).
\newblock Contribution to the theory of "student's" t-test as applied to the
  problem of two samples.
\newblock \textit{Statistical Research Memoirs} \textbf{2}, 1--24.

\bibitem[{Johansen(1980)}]{johansen1980welch}
\textsc{Johansen, S.} (1980).
\newblock The welch-james approximation to the distribution of the residual sum
  of squares in a weighted linear regression.
\newblock \textit{Biometrika} \textbf{67}, 85--92.

\bibitem[{Kim \& Cohen(1998)}]{kim1998behrens}
\textsc{Kim, S.-H.} \& \textsc{Cohen, A.~S.} (1998).
\newblock On the behrens-fisher problem: a review.
\newblock \textit{Journal of Educational and Behavioral Statistics}
  \textbf{23}, 356--377.

\bibitem[{Krishnamoorthy \& Yu(2004)}]{krishnamoorthy2004modified}
\textsc{Krishnamoorthy, K.} \& \textsc{Yu, J.} (2004).
\newblock Modified nel and van der merwe test for the multivariate
  behrens--fisher problem.
\newblock \textit{Statistics \& probability letters} \textbf{66}, 161--169.

\bibitem[{Lancaster(1964)}]{lancaster1964eigenvalues}
\textsc{Lancaster, P.} (1964).
\newblock On eigenvalues of matrices dependent on a parameter.
\newblock \textit{Numerische Mathematik} \textbf{6}, 377--387.

\bibitem[{Magnus(1985)}]{magnus1985differentiating}
\textsc{Magnus, J.~R.} (1985).
\newblock On differentiating eigenvalues and eigenvectors.
\newblock \textit{Econometric Theory} \textbf{1}, 179--191.

\bibitem[{Martin \& Liu(2015)}]{martin2014marginal}
\textsc{Martin, R.} \& \textsc{Liu, C.} (2015).
\newblock Marginal inferential models: Prior-free probabilistic inference on
  interest parameters.
\newblock \textit{Journal of the American Statistical Association}
  \textbf{110}, 1621--1631.

\bibitem[{Mickey \& Brown(1966)}]{mickey1966bounds}
\textsc{Mickey, M.~R.} \& \textsc{Brown, M.~B.} (1966).
\newblock Bounds on the distribution functions of the behrens-fisher statistic.
\newblock \textit{The Annals of Mathematical Statistics} \textbf{37}, 639--642.

\bibitem[{Nel \& Van~der Merwe(1986)}]{nel1986solution}
\textsc{Nel, D.} \& \textsc{Van~der Merwe, C.} (1986).
\newblock A solution to the multivariate behrens-fisher problem.
\newblock \textit{Communications in Statistics-Theory and Methods} \textbf{15},
  3719--3735.

\bibitem[{Nel et~al.(1990)Nel, van~der Merwe \& Moser}]{nel1990exact}
\textsc{Nel, D.~d.}, \textsc{van~der Merwe, C.~A.} \& \textsc{Moser, B.}
  (1990).
\newblock The exact distributions of the univariate and multivariate
  behrens-fisher statistics with a comparison of several solutions in the
  univariate case.
\newblock \textit{Communications in Statistics-Theory and Methods} \textbf{19},
  279--298.

\bibitem[{Olkin \& Marshall(2016)}]{olkin2016inequalities}
\textsc{Olkin, I.} \& \textsc{Marshall, A.~W.} (2016).
\newblock \textit{Inequalities: Theory of majorization and its applications},
  vol. 143.
\newblock Academic press.

\bibitem[{Pan et~al.(2013)Pan, Xu, Hu et~al.}]{pan2013some}
\textsc{Pan, X.}, \textsc{Xu, M.}, \textsc{Hu, T.} et~al. (2013).
\newblock Some inequalities of linear combinations of independent random
  variables: Ii.
\newblock \textit{Bernoulli} \textbf{19}, 1776--1789.

\bibitem[{Rao(1973)}]{rao1973linear}
\textsc{Rao, C.~R.} (1973).
\newblock \textit{Linear statistical inference and its applications}, vol.~2.
\newblock Wiley New York.

\bibitem[{Ruben(2002)}]{ruben2002simple}
\textsc{Ruben, H.} (2002).
\newblock A simple conservative and robust solution of the behrens-fisher
  problem.
\newblock \textit{Sankhy{\=a}: The Indian Journal of Statistics, Series A}
  \textbf{64}, 139--155.

\bibitem[{Satterthwaite(1946)}]{satterthwaite1946approximate}
\textsc{Satterthwaite, F.~E.} (1946).
\newblock An approximate distribution of estimates of variance components.
\newblock \textit{Biometrics bulletin} \textbf{2}, 110--114.

\bibitem[{Scheff{\'e}(1943)}]{scheffe1943solutions}
\textsc{Scheff{\'e}, H.} (1943).
\newblock On solutions of the behrens-fisher problem, based on the
  t-distribution.
\newblock \textit{The Annals of Mathematical Statistics} \textbf{14}, 35--44.

\bibitem[{Scheff{\'e}(1970)}]{scheffe1970practical}
\textsc{Scheff{\'e}, H.} (1970).
\newblock Practical solutions of the behrens-fisher problem.
\newblock \textit{Journal of the American Statistical Association} \textbf{65},
  1501--1508.

\bibitem[{Tsui \& Weerahandi(1989)}]{tsui1989generalized}
\textsc{Tsui, K.-W.} \& \textsc{Weerahandi, S.} (1989).
\newblock Generalized p-values in significance testing of hypotheses in the
  presence of nuisance parameters.
\newblock \textit{Journal of the American Statistical Association} \textbf{84},
  602--607.

\bibitem[{Welch(1947)}]{welch1947generalization}
\textsc{Welch, B.~L.} (1947).
\newblock The generalization ofstudent's' problem when several different
  population variances are involved.
\newblock \textit{Biometrika} \textbf{34}, 28--35.

\bibitem[{Yao(1965)}]{yao1965approximate}
\textsc{Yao, Y.} (1965).
\newblock An approximate degrees of freedom solution to the multivariate
  behrens fisher problem.
\newblock \textit{Biometrika} \textbf{52}, 139--147.

\end{thebibliography}

\end{document}